\documentclass[journal=jacsat,12pt]{achemso}
\usepackage[utf8]{inputenc}

\usepackage{amsthm}
\newtheorem{theorem}{Theorem}
\usepackage{amsmath,tikz}
\usepackage{atbegshi} 
\setkeys{acs}{articletitle = true}
\usepackage{chemformula}
\usepackage{algorithm}
\usepackage{algpseudocode}
\usepackage{braket}
\usepackage{calrsfs}
\usepackage{array}

\def\t{{\text{t}}}
\usepackage{amssymb}
\usepackage{float}
\usepackage{multirow}
\usepackage{pifont}
\usepackage[x11names,table]{xcolor}
\usepackage{subfigure}

\usepackage{bm}
\addtocounter{page}{0}
\usepackage{array}
\usepackage{mciteplus}

\begin{tocentry}
\includegraphics[width=8.25cm, height=4.5cm]{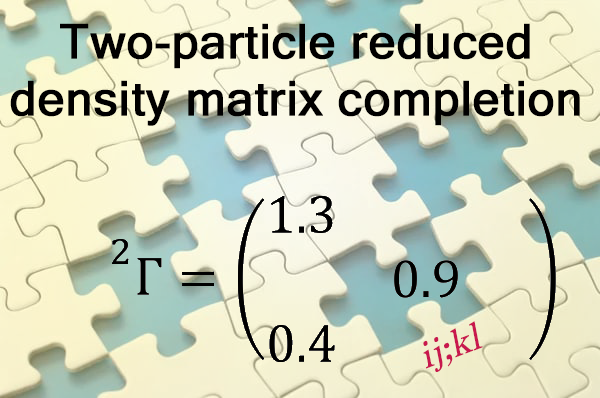}
\end{tocentry}

\title{Is the matrix completion of reduced density matrices unique?}

\author{Gustavo E. Massaccesi}
\affiliation{Departamento de Ciencias Exactas, Ciclo B\'asico Com\'un, Universidad de Buenos Aires, Ciudad Universitaria, 1428 Buenos Aires, Argentina and Instituto de Investigaciones Matem\'aticas ``Luis A. Santal\'o'' (IMAS), Consejo Nacional de Investigaciones Cient\'ificas y T\'ecnicas, Universidad de Buenos Aires. Ciudad Universitaria, 1428 Buenos Aires, Argentina}

\author{Ofelia B. O\~{n}a}
\affiliation{Instituto de Investigaciones Fisicoqu\'imicas Te\'oricas y Aplicadas, Universidad Nacional de La Plata, Consejo Nacional de Investigaciones Cient\'ificas y T\'ecnicas. Diag. 113 y 64 (S/N), Sucursal 4, CC 16, 1900 La Plata, Argentina}

\author{Luis Lain}
\affiliation{Department of Physical Chemistry, Faculty of Science and Technology, University of the Basque Country. PO Box 644, E-48080 Bilbao, Spain}

\author{Alicia Torre}
\affiliation{Department of Physical Chemistry, Faculty of Science and Technology, University of the Basque Country. PO Box 644, E-48080 Bilbao, Spain}

\author{Juan E. Peralta}
\email{juan.peralta@cmich.edu}
\affiliation{Department of Physics, Central Michigan University, Mount Pleasant, MI, 48859, USA}

\author{Diego R. Alcoba}
\email{dalcoba@df.uba.ar}
\affiliation{Universidad de Buenos Aires, Facultad de Ciencias Exactas y Naturales, Departamento de F\'isica. Ciudad Universitaria, 1428 Buenos Aires, Argentina and CONICET - Universidad de Buenos Aires, Instituto de F\'isica de Buenos Aires (IFIBA). Ciudad Universitaria, 1428 Buenos Aires, Argentina}

\author{Gustavo E. Scuseria}
\email{guscus@rice.edu}
\affiliation{Department of Chemistry, Rice University, Houston, TX 77005-1892}
\alsoaffiliation{Department of Physics and Astronomy, Rice University, Houston, TX 77005-1892}


\begin{document}

\begin{abstract}

Reduced density matrices are central to describing observables in many-body quantum systems. 
In electronic structure theory, the two-particle reduced density matrix (2-RDM) suffices to determine the energy and other key properties. Recent work has used matrix completion, leveraging the low-rank structure of RDMs and approximate theoretical models, to reconstruct the 2-RDM from partial data and thus reduce computational cost. 
However, matrix completion is, in  general, an under-determined problem. 
Revisiting Rosina’s theorem [M. Rosina, Queen's Papers on Pure and Applied Mathematics No. 11, 369 (1968)], 
we here show that the matrix completion is unique under certain conditions, 
identifying the subset of 2-RDM elements that enables its exact reconstruction from incomplete information.
Building on this, we introduce a hybrid quantum–stochastic algorithm that achieves exact matrix completion, demonstrated through applications to the Fermi–Hubbard model.

\end{abstract}

\newpage

The two-particle reduced density matrix (2-RDM) corresponding to an $N$-particle wave function $\Psi$ is defined by\cite{Coleman.book.2000,Davidson1976Reduced}
\begin{equation}
    ^2\Gamma(x_1,x_2;x'_1,x'_2)     
    = \int \; ^N\Gamma(x_1,x_2,x_3,\dots,x_N;x'_1,x'_2,x_3,\dots,x_N) \; dx_3 \dots dx_N
\label{eq1}
\end{equation}
where $^N\Gamma(\bold{x};\bold{x'})=\Psi(\bold{x})\Psi^*(\bold{x'})$ is the $N$-particle density matrix, and normalization has been omitted for clarity. 
This matrix carries all the relevant information needed to evaluate expectation values of one- and two-particle operators, as it is often the  case. 
For example, the energy for a pairwise interacting $N$-particle system may be exactly written as\cite{Coleman.RevModPhys.1963}
\begin{equation}
    E_\Psi=\sum_{ij;kl} \mathrm{^2H}_{ij;kl} \; ^2\Gamma_{ij;kl} \,,
\label{eq2}
\end{equation}
where $\{\mathrm{^2H}_{ij;kl}\}$ is the two-particle reduced Hamiltonian, 
\begin{equation}
\begin{split}
    ^2\Gamma_{ij;kl} &\equiv \int \varphi^*_i(x_1) \varphi^*_j(x_2) \; ^2\Gamma(x_1,x_2;x'_1,x'_2)  \;\varphi_k(x'_1) \varphi_l(x'_2) \; dx_1 dx_2 dx'_1 dx'_2 
    \\
    &= \langle \Psi | {\hat a}^{\dagger}_i {\hat a}^{\dagger}_j {\hat a}_l {\hat a}_k | \Psi \rangle,
\end{split}
\label{eq3}
\end{equation}
and ${\hat a}^\dagger_{i}$ and ${\hat a}_{i}$ denote particle creation and annihilation operators acting on a finite-size single-particle basis $\{\varphi_i\}$, respectively. 
The 2-RDM is a much more compact and economic storage of information than the $N$-particle wave function.\cite{Coleman.book.2000, Mazziotti.book.2007}
However, while the wave function $\Psi$ must satisfy appropriate exchange symmetry and normalization, the 2-RDM must satisfy the so-called  $N$-representability conditions,\cite{Coleman.RevModPhys.1963,Garrod.JMP.1963,Kummer.JMP.1967,Erdahl.IJQC.1978,Mazziotti.PRL.2012} which bear a complexity that grows exponentially with $N$, in order to fulfill Eq. (\ref{eq1}).\cite{Mazziotti.PRL.2012,nielsen2010quantum,Liu.PRL.2007}

Eq. (\ref{eq1}) provides a simple prescription to calculate the 2-RDM given a preimage $\Psi$. The inverse problem, i.e., given a 2-RDM, how to derive a preimage $\Psi$, is known as the reconstruction problem. 
This problem lies at the heart of quantum state and process tomography, which are fundamental techniques used to characterize 
unknown quantum states and processes as well as to quantify the quality of quantum devices.\cite{nielsen2010quantum,PhysRevLett.93.080502,Xin2017Quantum} 
As it will be shown in this Letter, the reconstruction problem is intimately related to the completion problem, which refers to the  process of completing the full, physically valid 2-RDM from a partial subset of its elements (typically those that are directly measured, or approximated).

In 1968, Rosina proved a theorem that establishes
conditions such that the reconstruction is unique.\cite{Rosina1968}  
The theorem shows that the 2-RDM corresponding to a non-degenerate ground state of a quantum system completely determines the exact $N$-particle wave function without any specific information about the
Hamiltonian other than bearing at most two-particle interactions.\cite{Mazziotti1998Contracted,Rosina2000Theorems}
In this Letter we demonstrate that the subset of elements of the 2-RDM needed for a unique reconstruction  of the preimage $\Psi$, and hence for the matrix completion  of the full 2-RDM,  
is linked to the subset of non-zero elements of the two-particle reduced Hamiltonian $\{\mathrm{^2H}_{ij;kl}\}$.

\begin{theorem}
\label{teo1}
(``Uniqueness RDM Completion Theorem") 
For an $N$-particle Hamiltonian with at most two-particle interactions  with reduced representation $\{\mathrm{^2H}_{ij;kl}\}$ and a non-degenerate ground state, the subset of elements of the  2-RDM corresponding to the ground state associated with the non-zero elements of $\{\mathrm{^2H}_{ij;kl}\}$ has a unique preimage, leading to a unique matrix completion of the full 2-RDM via Eq. (\ref{eq1}).
\end{theorem}

\begin{proof}
Let $S$ be the subset of indices for which the elements of the two-particle reduced Hamiltonian $\{\mathrm{^2H}_{ij;kl}\}$ are non-zero. The energy, Eq. (\ref{eq2}), may be calculated as follows:
\begin{equation}
    E_\Psi=\sum_{(ij;kl)\;\in S}\; \mathrm{^2H}_{ij;kl} \; ^2\Gamma_{ij;kl} \,.
\end{equation}
Since the Hamiltonian consists solely of two-particle operators, the ground-state energy is completely determined by the 2-RDM elements in $S$ and is minimized by the exact ground-state solution. If there existed another $N$-particle-density matrix $^N\Gamma$
 yielding the same 2-RDM elements in $S$ via Eq. (\ref{eq1}), it would necessarily produce the same energy and thus also correspond to a ground state (or an ensemble of ground states), as only ground states attain the minimal energy. 
 This would contradict the assumed non-degeneracy of the ground state. 
  Finally, like in Rosina's theorem,\cite{Rosina1968,Rosina2000Theorems,Mazziotti1998Contracted,MAZZIOTTI2000212,Mazziotti_book2007} since a nondegenerate ground state cannot be represented as a nontrivial ensemble, 
 such a subset of elements of the 2-RDM must admit a pure-state preimage. In consequence, this leads to a {\em unique} matrix completion of the 2-RDM via Eq. (\ref{eq1}).
\end{proof}

We note that only the location of the subset $S$ in the 2-particle reduced Hamiltonian is needed for completion, not the actual matrix element values.
Moreover,  the subset $S$ is basis-dependent, and thus the number of matrix elements needed for the completion 
varies with the representation. 
 To illustrate how the Theorem manifests in practical applications, 
we utilize a hybrid quantum–stochastic algorithm that numerically performs the matrix completion of a 2-RDM.  
The algorithm consists of applying to an initial $N$-particle density matrix a sequence of unitary evolution operators generated by a stochastic process that iteratively refines the partial information encoded in the reduced two-particle state. A similar algorithm has been employed by us to numerically determine the $N$-representability of a pure RDM.\cite{Massaccesi__jctc_2024} 
Through this iterative procedure, the elements of the 2-RDM converge toward the critical subset of elements of a target 2-RDM associated with the non-degenerate ground state of an $N$-particle Hamiltonian involving at most two-particle interactions, thereby enabling its exact completion. 
A summary of the procedure is shown in Algorithm~\ref{Annealing}.

\begin{figure}
\begin{algorithm}[H]
\caption{
Hybrid quantum–stochastic reduced density matrix completion algorithm}
\label{Annealing}
\begin{algorithmic}
\Require 
Subset of elements of target $^2\Gamma_\t$, initial state $\ket{\Psi_0}$, operator pool $\{\hat{O}_i\}_{i=1}^{M}$,\cite{Evangelista.JCP.2019} 
initial temperature $T_0$, temperature decay $\delta_T$, 
initial maximum angle $\theta_{\max}^{(0)}$, and angle update 
$\delta_{\theta}^{\pm}$.
\Ensure $\mathcal{D}_{\min}$.
\Procedure{}{}
\State Initialize $\mathcal{D}_0 = \mathcal{D}(\Psi_0)$.
\State Set $k \gets 0$.
\While{$k < k_{\max}$}
    \State Randomly select an operator $\hat{O}$ from the operator pool $\{\hat{O}_i\}$.
    \State Sample random rotation angle $\theta \sim \mathcal{U}[-\theta_{\max}^{(k)},\,\theta_{\max}^{(k)}]$.
    \State Generate candidate state $\ket{\Psi'} = e^{\theta \hat{O}} \ket{\Psi_k}$.
    \State Compute Hilbert–Schmidt distance to target (cost function) $\mathcal{D}' = \mathcal{D}(\Psi')$
    \State Compute $\Delta \mathcal{D} = \mathcal{D}' - \mathcal{D}_k$.
    \If{$\Delta \mathcal{D} \le 0$}
        \State Accept new configuration: $\ket{\Psi_{k+1}} \gets \ket{\Psi'}$, $\mathcal{D}_{k+1} \gets \mathcal{D}'$.
        \State Update $\theta_{\max}^{(k+1)} = \delta_{\theta}^{+} \, \theta_{\max}^{(k)}$.
    \Else
        \State Compute acceptance probability $p = \exp(-\Delta \mathcal{D} / T_k)$.
        \If{random$(0,1) < p$}
            \State Accept configuration: $\ket{\Psi_{k+1}} \gets \ket{\Psi'}$, $\mathcal{D}_{k+1} \gets \mathcal{D}'$.
            \State Update $\theta_{\max}^{(k+1)} = \delta_{\theta}^{+} \, \theta_{\max}^{(k)}$.
        \Else
            \State Reject configuration: $\ket{\Psi_{k+1}} \gets \ket{\Psi_k}$, $\mathcal{D}_{k+1} \gets \mathcal{D}_k$.
            \State Update $\theta_{\max}^{(k+1)} = \delta_{\theta}^{-} \, \theta_{\max}^{(k)}$.
        \EndIf
    \EndIf
    \State $T_{k+1} = \delta_T \, T_k$.
    \State $k \gets k + 1$.
\EndWhile
\State Record $\mathcal{D}_{\min} = \min_k \mathcal{D}_k$.
\EndProcedure
\end{algorithmic}
\end{algorithm}
\end{figure}

As a proof-of-concept, we consider the ground state of the inhomogeneous Fermi–Hubbard model
for a three-site one-dimensional lattice with open boundary conditions at half-filling, defined by the Hamiltonian\cite{Hubbard1963Electron}
\begin{equation}
\label{eq:hubbard} {\hat H }= 
-\sum_{\langle i, j \rangle,\sigma} ({\hat a}_{i\sigma}^\dagger {\hat a}_{j\sigma} + {\hat a}_{j\sigma}^\dagger {\hat a}_{i\sigma}) 
+\sum_{i,\sigma} \epsilon_{i\sigma} {\hat n}_{i\sigma}
+ U \sum_i {\hat n}_{i\uparrow}{\hat n}_{i\downarrow}, 
\end{equation}
where ${\hat a}^\dagger_{i\sigma}$ $({\hat a}_{i\sigma})$ 
denotes the fermionic creation (annihilation) operator for a particle at site $i$
with spin $\sigma$ ($\uparrow, \downarrow$), and ${\hat n}_{i\sigma}={\hat a}^\dagger_{i\sigma}{\hat a}_{i\sigma}$ is the corresponding number operator. 
The notation $\langle i, j \rangle$ indicates nearest-neighbor pairs on the lattice. 
The on-site energy parameters $\epsilon_i$ are introduced to break spin and spatial symmetries and remove ground-state degeneracies,
and the on-site interaction strength $U$ is assumed to be positive.

\begin{figure}
\centering
\includegraphics[width=1.0\linewidth]{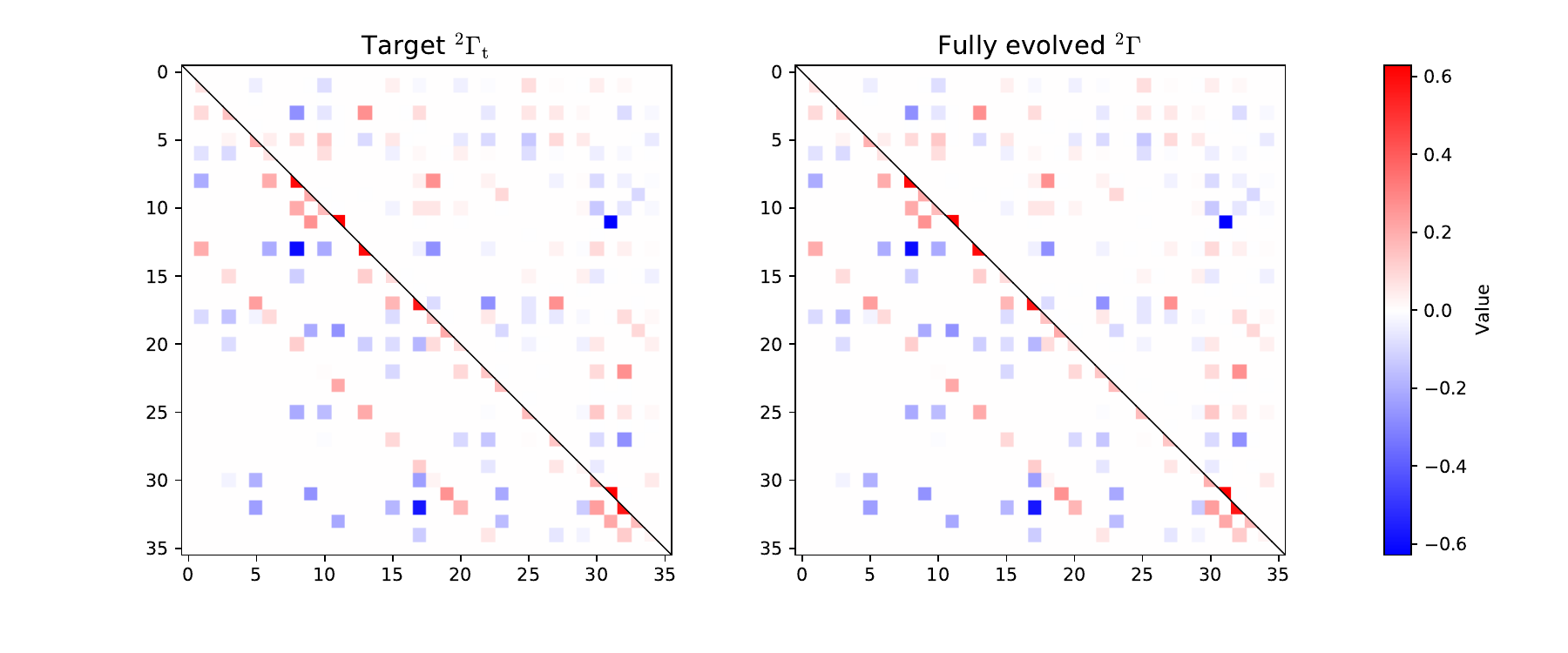}

\caption{Heat-map representations in the lattice-site basis of the target 2-RDM (left panel), and the reconstructed (fully evolved) density-matrix elements obtained via matrix completion (right panel)
corresponding to the non-degenerate ground state of the inhomogeneous three-site one-dimensional Fermi–Hubbard model with open boundary conditions at half-filling. Density-matrix elements associated with the symmetrized non-zero (zero) elements of the corresponding two-particle reduced Hamiltonian are shown as the lower (upper) triangular part of the matrices.}
\label{fig1}
\end{figure}

In the lattice-site basis, the target 2-RDM corresponding to the non-degenerate ground state comprises 900 elements, out of which 360 fulfill $S_z$-spin symmetry and are non-zero. 
Among these, 184 elements correspond to the symmetrized non-zero elements of the two-particle reduced Hamiltonian, while the remaining elements are determined under the   conditions established in Theorem~\ref{teo1}, as illustrated in Figure~\ref{fig1}. 
The left panel of
Figure~\ref{fig2} shows  the evolution of the partial (critical subset of elements) and complete (full set of elements) Hilbert–Schmidt distances between the  reduced two-particle state of the unitarily evolved $N$-particle density matrix and the target 2-RDM, respectively, obtained from the completion algorithm. 
The corresponding energy deviation and infidelity of the evolved $N$-particle density matrix from the exact ground state are also presented in the right panel of Figure~\ref{fig2}. The evolution is initiated from an arbitrary linear combination of the ground and first-excited states. Overall, the numerical results shown in Figures~\ref{fig1} and ~\ref{fig2} demonstrate that the elements of the evolved 2-RDM progressively converge toward both the critical subset and the remaining components of the target 2-RDM associated with the non-degenerate ground state, thereby demonstrating exact completion.
An analogous behavior is observed in the eigenbasis of the two-particle reduced Hamiltonian. In this representation, 27 elements of the target 2-RDM correspond to the symmetrized non-zero reduced two-particle Hamiltonian entries, while the remaining elements are determined by the conditions stated in Theorem~\ref{teo1}, as illustrated in Figure~\ref{fig:diag}.
As in the lattice case, evolution from a generic low-energy superposition leads to convergence to the exact ground-state 2-RDM, confirming the robustness and generality of the matrix completion scheme across different representations.
Choosing an eigenbasis representation may be particularly useful in systems where symmetry cannot be exploited, such as asymmetric molecules.


\begin{figure}
\centering
\includegraphics[width=1.0\linewidth]{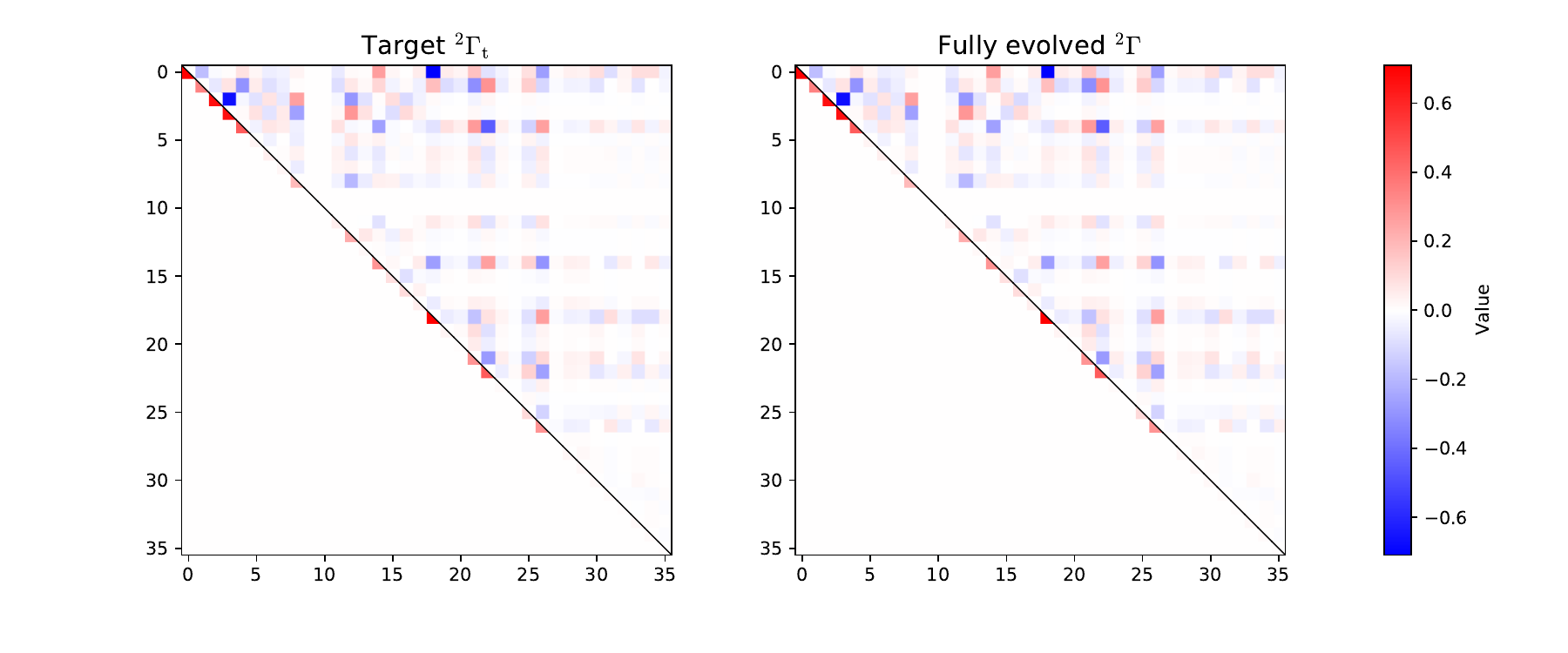}
\caption{Heat-map representations in the eigenbasis of the two-particle reduced Hamiltonian of the target 2-RDM (left panel), and the reconstructed (fully evolved) density-matrix elements obtained via matrix completion (right panel) corresponding to the non-degenerate ground state of the inhomogeneous three-site one-dimensional Fermi–Hubbard model with open boundary conditions at half-filling. Density-matrix elements associated with the symmetrized non-zero (zero) elements of the corresponding two-particle reduced Hamiltonian are shown as the lower (upper) triangular part of the matrices.}
\label{fig:diag}
\end{figure}

Thus far, we have evaluated our algorithm in the noiseless matrix completion setting, where the partial information consists of the critical subset of 2-RDM elements associated with a non-degenerate ground state. To assess the robustness of the methodology under more challenging conditions, we introduce statistical noise into the target 2-RDM in the lattice-site basis, $^2\Gamma_\t$, while retaining the same subset of known elements. Specifically, we define
\begin{equation}
\label{eq:noise}
^2\Gamma_\t(\varepsilon) = {}^2\Gamma_\t + \varepsilon R ,
\end{equation}
where $R$ has elements drawn from a uniform distribution in $[-1,1]$ and the parameter $\varepsilon \in [0,0.1]$ controls the noise strength. 
The chosen values of $\varepsilon$ produce noticeable distortions in the target 2-RDM.
In this regime, the algorithm is expected to converge toward the target only up to a noise-dependent limit: larger noise strengths should yield larger $\mathcal{D}_{\min}$, while in the noiseless case $\mathcal{D}_{\min} \to 0$, as previously demonstrated.
Table~\ref{tab:noise1}  confirms this expectation, showing that 
the converged Hilbert-Schmidt distances for $^2\Gamma_\t(\varepsilon)$ increase as the noise strength increases. 
These results emphasize applications for the proposed algorithm: It can be used not only in noiseless matrix completion settings but also in noisy ones, constructing a $2$-RDM (the evolved RDM) that is closest to the target.

Rosina’s theorem  proves that the 2-RDM corresponding to a non-degenerate  
ground state of a quantum system completely determines the exact $N$-particle wave function without any specific information about the Hamiltonian other than bearing at most two-particle interactions. 
Building on that theorem, we derive the proof that rigorous matrix completion uniqueness conditions exist.
The subset of elements of the 2-RDM needed for the determination of
the unique preimage $\Psi$, and hence for a full reconstruction of the 2-RDM, is linked to the subset of
non-zero elements of the two-particle reduced Hamiltonian.
This result establishes clear conditions under which the 2-RDM can be uniquely reconstructed from incomplete information, thereby strengthening the theoretical underpinnings of matrix completion in electronic structure theory. 
Moreover, the formal conditions established here support future developments in quantum tomography\cite{PhysRevLett.118.020401} by identifying the minimal information needed to reconstruct physically valid two-particle correlations. 
The framework also provides a systematic way to correct defective or noisy RDMs, making it relevant for error mitigation techniques\cite{10.1063/1.4994618} on near term quantum devices,\cite{Rubin_2018, Smart.pra.2022} and complements recent efforts to develop  matrix completion strategies for fermionic RDMs.\cite{Peng2023Fermionic}
As a numerical proof-of-concept, we  have here demonstrated its applicability 
using  a hybrid quantum–stochastic algorithm that achieves unique matrix completion of the full 2-RDM of a Fermi–Hubbard model from a subset of it. 
%
%

\begin{figure}
\centering
\includegraphics[width=1.0\linewidth]{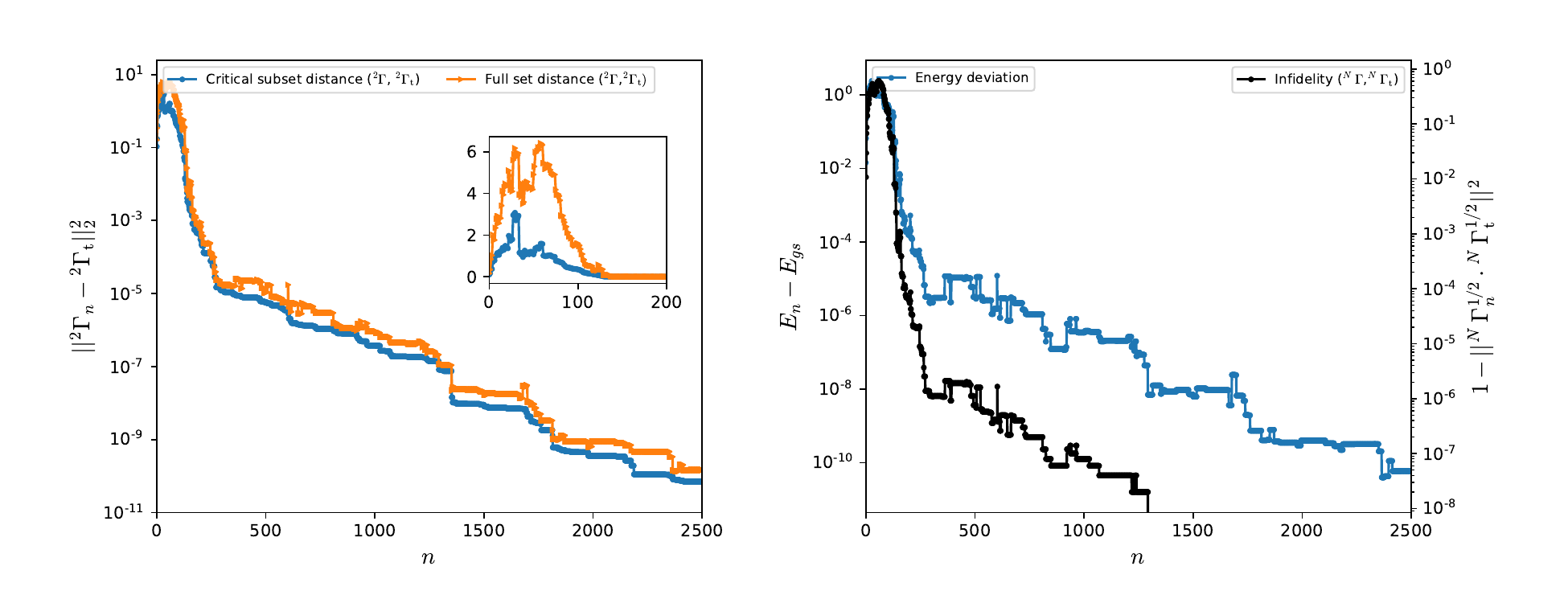}
\caption{Left panel: 
Partial (critical subset of elements) and complete (full set of elements) Hilbert–Schmidt distances between the  reduced two-particle state of the unitarily evolved $N$-particle density matrix, $^2\Gamma$, and the target 2-RDM, $^2\Gamma_\t$, respectively, as a function of iteration number $n$, for the exact ground state of the model system in Figure~\ref{fig1}.
The evolution starts from an arbitrary linear combination of the ground and first-excited states. Inset: enlarged view of the first 200 iterations.  
Right panel: Energy deviation and infidelity of the evolved $N$-particle density matrix with respect to the exact ground-state density matrix as a function of iteration number $n$.
\label{fig2} }
\end{figure}

\begin{table}[H]
\centering
\renewcommand{\arraystretch}{1.4}
\caption{
Partial (critical subset of elements) and complete (full set of elements) Hilbert–Schmidt distances between the  reduced two-particle state of the unitarily evolved $N$-particle density matrix, $^2\Gamma$, and the noiseless and noisy target 2-RDMs, $^2\Gamma_\t$ and $^2\Gamma_\t(\varepsilon)$, respectively, for the exact ground state of the model system in Figure~\ref{fig1}. The noisy targets are constructed by adding random noise of strength $\varepsilon$ to the noiseless 2-RDM (see text for details).}
\label{tab:noise1}
\vspace{5mm}
\begin{tabular}{|c|c|c|c|}
\hline
\multirow{3}{*}{$\varepsilon$} & \multicolumn{3}{c|}{Hilbert–Schmidt distance} \\
\cline{2-4}
                                & \multicolumn{2}{c|}{  Critical subset   } & \multicolumn{1}{c|}{Full set} \\
\cline{2-4}
                                & ($^2\Gamma,^2\Gamma_\t(\varepsilon)$) & $(^2\Gamma,^2\Gamma_\t)$  &  $(^2\Gamma,^2\Gamma_\t)$\\
\hline
$10^{-3}$	 & $2.99 \times \;10^{-5}$	  & $2.44  \times\;10^{-5}$	& $1.28 \times\;10^{-4}$\\
\hline
$10^{-2}$	 & $3.22 \times \;10^{-3}$ &  $2.44 \times \;10^{-3}$	& $1.52 \times\;10^{-2}$\\
\hline
$10^{-1}$	     & $5.94 \times \;10^{-1}$ &  $2.52 \times\;10^{-1}$  & $4.08 \times\;10^{-1}$\\
\hline
\end{tabular}
\end{table}

\section{Acknowledgments}

OBO, GEM, PC, and DRA acknowledge the financial support from the Consejo Nacional de Investigaciones Cient\'{\i}ficas y T\'ecnicas (grant No. PIP KE3 11220200100467CO and PIP KE3 11220210100821CO). OBO, GEM, PC, and DRA acknowledge support from the Universidad de Buenos Aires (grant No. 20020190100214BA and 20020220100069BA) and the Agencia Nacional de Promoci\'on Cient\'{\i}fica y Tecnol\'ogica (grant No. PICT-201-0381).  JEP acknowledges support from NSF award number DMR-2318872. 
GES is a Welch Foundation Chair (C-0036), and his work was supported by the U.S. Department of Energy under Award No. DE-SC0019374. 

\bibliography{references}

\end{document}